\documentclass[a4paper,11pt]{article}

\usepackage{fullpage}
\usepackage{times}
\usepackage{soul}
\usepackage{url}
\usepackage[utf8]{inputenc}
\usepackage[small]{caption}
\usepackage{graphicx}
\usepackage{xcolor}
\usepackage{amsmath}
\usepackage{booktabs}
\usepackage{algorithm}
\usepackage{enumitem}

\usepackage{amsthm}
\usepackage{amssymb}
\usepackage{algpseudocode}
\usepackage{cleveref}

\usepackage{bbm}

\newtheorem{theorem}{Theorem}

\newtheorem{lemma}[theorem]{Lemma}

\theoremstyle{definition}

\algnotext{EndFor}
\algnotext{EndIf}
\algnotext{EndWhile}
\algnotext{EndProcedure}

\usepackage{natbib}

\allowdisplaybreaks

\usepackage{authblk}

\title{\bf A Characterization of Maximum Nash Welfare \\for Indivisible Goods}

\author{Warut Suksompong}
\affil{National University of Singapore}

\date{\vspace{-10mm}}

\begin{document}

\maketitle

\begin{abstract}
In the allocation of indivisible goods, the maximum Nash welfare (MNW) rule, which chooses an allocation maximizing the product of the agents' utilities, has received substantial attention for its fairness.
We characterize MNW as the only additive welfarist rule that satisfies envy-freeness up to one good.
Our characterization holds even in the simplest setting of two agents.
\end{abstract}

\section{Introduction} \label{sec:intro}

The fair allocation of limited resources among interested parties---often referred to as \emph{fair division}---is a fundamental problem in economics \citep{Moulin03}.
Its applications range from inheritance division to budget distribution to divorce settlement.
While much of the early work in fair division dealt with the case of \emph{divisible} resources such as time or land, a significant portion of the recent literature has focused on the allocation of \emph{indivisible} goods such as artwork, jewelry, and electronic devices \citep{BouveretChMa16,AmanatidisAzBi22}.

Which method should one use to allocate indivisible goods fairly?
Among numerous possible methods that one may employ, \citet{CaragiannisKuMo19} proposed using the \emph{maximum Nash welfare (MNW)} rule, which chooses an allocation that maximizes the product of the agents' utilities, or equivalently, the sum of their logarithms.
These authors demonstrated the ``unreasonable fairness'' of MNW: the allocation output by this rule always satisfies \emph{envy-freeness up to one good (EF1)}---any envy that an agent may have toward another agent can be eliminated by removing some good in the latter agent's bundle---along with the economic efficiency notion of \emph{Pareto optimality (PO)}---no reallocation of the goods makes at least one agent better off and no agent worse off.\footnote{This result relies on a tie-breaking specification when it is impossible to give every agent nonzero utility.}
Even though this result admits a rather simple proof, it is arguably one of the most important results in the theory of fair division, as it has inspired many further investigations of fairness with indivisible resources.\footnote{As of December 2022, the paper by \citet{CaragiannisKuMo19} has received over 450 citations according to Google Scholar.}

While MNW offers compelling guarantees in EF1 and PO, it is far from being the only rule to do so.
For instance, \citet{BarmanKrVa18} devised another procedure that also ensures the same two properties; their procedure has a more complex description but a better running time guarantee than MNW.
More generally, since the definitions of EF1 and PO only concern individual profiles and do not relate different profiles (unlike, e.g., strategyproofness), one could define a large number of rules that return an EF1 and PO allocation for every profile.
This raises the question of whether MNW is the ``fairest'' within some meaningful set of rules.
A natural class of rules is that of \emph{additive welfarist rules}, which choose an allocation maximizing a welfare notion that can be expressed as the sum of some increasing function of the agents' utilities \citep[p.~67]{Moulin03}.
MNW corresponds to the case of the logarithm function, whereas the \emph{maximum utilitarian welfare (MUW)} rule---another commonly studied rule---results from taking the identity function.
By definition, additive welfarist rules produce PO allocations, because a Pareto improvement would necessarily lead to a higher welfare.\footnote{One needs to be careful with tie-breaking when the welfare is (negative) infinity, but this is not our focus.}

The purpose of this note is to characterize MNW as the \emph{only} additive welfarist rule that guarantees EF1, thereby establishing it as the ``fairest'' rule within this class.
To the best of our knowledge, this is the first characterization of MNW in the setting of indivisible goods.\footnote{Analogues of MNW have been characterized within the class of additive welfarist rules in the context of randomized collective choice \citep[Sec.~6]{BogomolnaiaMoSt02}, cake cutting \citep[Sec.~4.5]{SegalhaleviSz19}, and cake sharing \citep[App.~A]{BeiLuSu22}. 
\citet{FreemanShVa20} characterized a variant of MNW for fractional allocations.}

\section{Preliminaries}

Let $N = \{1,\dots,n\}$ be the set of agents and $G = \{g_1,\dots,g_m\}$ be the set of goods.
Each agent~$i\in N$ has a utility function $u_i \colon 2^G\rightarrow\mathbb{R}_{\ge 0}$; we write $u_i(g)$ instead of $u_i(\{g\})$ for a single good $g\in G$.
We assume that the utility functions are additive,\footnote{Without additivity, an allocation produced by MNW is not necessarily EF1 \citep[App.~C]{CaragiannisKuMo19}.} that is, $u_i(G') = \sum_{g\in G'}u_i(g)$ for all $i\in N$ and $G'\subseteq G$.
A \emph{profile} consists of $N$, $G$, and $(u_i)_{i\in N}$.
An \emph{allocation} $A = (A_1,\dots,A_n)$ is an ordered partition of $G$ into $n$~bundles such that bundle~$A_i$ is allocated to agent~$i$.
An allocation is \emph{EF1} if for all pairs $i,j\in N$ such that $A_j\neq\emptyset$, there exists a good $g\in A_j$ with the property that $u_i(A_i) \ge u_i(A_j\setminus\{g\})$.
A \emph{rule} maps any given profile to an allocation.

Let $f\colon [0,\infty)\rightarrow [-\infty, \infty)$ be an increasing function,\footnote{Note that even if we were to include $\infty$ in the codomain of $f$, it cannot be in the actual range because $f$ is increasing.}
and assume that $f$ is differentiable on $(0,\infty)$.
Given any profile, an \emph{additive welfarist rule with function $f$} chooses an allocation~$A$ that maximizes the welfare $\sum_{i\in N}f(u_i(A_i))$.
If there are multiple such allocations, the rule may choose one arbitrarily; the choice of tie-breaking will not matter for our result.
As mentioned earlier, MNW corresponds to taking $f(x) = \ln x$ (or, more generally, $f(x) = a\ln x + b$ for some constants $a > 0$ and $b\in\mathbb{R}$), while MUW corresponds to taking $f(x) = x$ (or, more generally, $f(x) = ax + b$ for some constants $a > 0$ and $b\in\mathbb{R}$).

\section{The Result}

We now state our characterization of MNW, which holds even in the simplest setting of two agents.

\begin{theorem}
\label{thm:main}
Suppose that an additive welfarist rule with function $f$ returns an EF1 allocation for every profile with $n = 2$ agents in which it is possible to give both agents nonzero utility.
Then, there exist constants $a > 0$ and $b\in\mathbb{R}$ such that $f(x) = a\ln x + b$ for all $x\ge 0$.
\end{theorem}

To establish \Cref{thm:main}, we will make use of the following auxiliary lemma.

\begin{lemma}
\label{lem:helper}
Let $f\colon [0,\infty)\rightarrow [-\infty, \infty)$ be an increasing function that is differentiable on $(0,\infty)$, and assume that for every positive integer~$k$, the function $h_k(x) := f((k+1)x) - f(kx)$ is constant on $(0,\infty)$.
Then, there exist constants $a > 0$ and $b\in\mathbb{R}$ such that $f(x) = a\ln x + b$  for all $x\ge 0$.
\end{lemma}

\begin{proof}
Suppose that $f$ satisfies the condition in the lemma statement.
It suffices to show that there exists a constant $a > 0$ such that $f'(x) = a/x$ for all $x > 0$; once this is established, $f(0) = -\infty$ follows because $f$ is increasing.
Moreover, since $f'(x) > 0$ for all $x > 0$, this is equivalent to showing that $yf'(y) = zf'(z)$ for all $y,z > 0$.

Consider any positive integer~$k$.
Since the function $h_k(x)$ is constant on $(0,\infty)$, so is the function $\widehat{h}_k(x) := f\left(\left(1 + \frac{1}{k}\right)x\right) - f(x) = h_k(x/k)$.
Let $c_k > 0$ be the constant such that $\widehat{h}_k(x) = c_k$ for all $x > 0$.
Now, fix any $y,z > 0$.
We have
\begin{align*}
yf'(y) 
&= y\cdot\lim_{\varepsilon\rightarrow 0}\frac{f(y+\varepsilon) - f(y)}{\varepsilon} \\
&= y\cdot\lim_{k\rightarrow \infty}\frac{f(y+\frac{y}{k}) - f(y)}{y/k} \\
&= \lim_{k\rightarrow \infty}\frac{f(y+\frac{y}{k}) - f(y)}{1/k} = \lim_{k\rightarrow \infty}\frac{f\left(\left(1 + \frac{1}{k}\right)y\right) - f(y)}{1/k} = \lim_{k\rightarrow\infty} kc_k.
\end{align*}
Similarly, $zf'(z) = \lim_{k\rightarrow\infty} kc_k$.
We therefore conclude that $yf'(y) = zf'(z)$, as required.
\end{proof}

With \Cref{lem:helper} in hand, we are now ready to prove \Cref{thm:main}.

\begin{proof}[Proof of \Cref{thm:main}]
Suppose that an additive welfarist rule~$R$ with function $f$ returns an EF1 allocation for every profile with two agents in which it is possible to give both agents nonzero utility.
By \Cref{lem:helper}, it is sufficient to show that for every positive integer~$k$, the function $h_k(x) := f((k+1)x) - f(kx)$ is constant on $(0,\infty)$.

Assume for the sake of contradiction that $h_k(x)$ is not constant on $(0,\infty)$ for some~$k$.
This means that there exist $y,z > 0$ such that $h_k(y)\ne h_k(z)$.
Without loss of generality, suppose that $h_k(y) > h_k(z)$, that is, $f((k+1)y) - f(ky) > f((k+1)z) - f(kz)$.
Since $f$ is differentiable, and therefore continuous, there exists a value $\varepsilon\in (0,z)$ with the property that
\begin{align}
\label{eq:difference}
f((k+1)y) - f(ky) > f((k+1)z-\varepsilon) - f(kz-\varepsilon). \tag{*}
\end{align}

Now, consider a profile with $n = 2$ agents and $m = 2k+1$ goods such that
\begin{itemize}
\item $u_1(g_1) = 0$ and $u_1(g_j) = y$ for each $j\in\{2,\dots,m\}$;
\item $u_2(g_1) = z-\varepsilon$ and $u_2(g_j) = z$ for each $j\in\{2,\dots,m\}$.
\end{itemize}
Clearly, it is possible to give both agents nonzero utility in this profile.
Since $f$ is increasing and $z-\varepsilon > 0$, our additive welfarist rule~$R$ must allocate $g_1$ to agent~$2$.
Given this, if $R$ allocates at most $k-1$ goods to agent~$1$, then the EF1 condition is violated for agent~$1$.
On the other hand, if $R$ allocates at least $k+1$ goods to agent~$1$, then the EF1 condition is violated for agent~$2$.
Hence, $R$ must allocate exactly $k$ goods to agent~$1$.
In particular, doing so must yield at least as high welfare as allocating exactly $k+1$ goods to agent~$1$.
It follows that
\begin{align*}
f(ky) + f(kz + (z-\varepsilon)) \ge f((k+1)y) + f((k-1)z + (z-\varepsilon)).
\end{align*}
That is,
\begin{align*}
f(ky) + f((k+1)z-\varepsilon) \ge f((k+1)y) + f(kz-\varepsilon),
\end{align*}
or equivalently,
\begin{align*}
f((k+1)z-\varepsilon) - f(kz-\varepsilon) \ge f((k+1)y) - f(ky).
\end{align*}
However, this is a contradiction to \eqref{eq:difference}.
\end{proof}

\Cref{thm:main} can be extended to any fixed number $n\ge 3$ of agents, by adding $n-2$ extra agents and $n-2$ extra goods.
Each extra agent has utility~$1$ for a distinct extra good and $0$ for the remaining goods, while each original agent has utility~$0$ for all extra goods.
A similar argument can then be applied to establish the characterization.

\subsection*{Acknowledgments}

The author acknowledges support from the Singapore Ministry of Education under grant number MOE-T2EP20221-0001 and from an NUS Start-up Grant, and thanks the anonymous reviewer for helpful comments.

\bibliographystyle{plainnat}
\bibliography{main}

\end{document}